\documentclass[11pt,a4paper]{article}

\usepackage{titlesec}
\usepackage{amssymb}
\usepackage{amsmath,amsthm,amssymb,color}
\usepackage[pdftex,pagebackref,colorlinks]{hyperref}
\usepackage{graphicx}
\usepackage{multirow}
\usepackage{makecell}
\usepackage{booktabs}
\usepackage{enumerate}
\usepackage{cleveref}
\date{}
\usepackage{url}
\usepackage{enumitem}
\usepackage{pifont}
\urldef{\mailsa}\path|{alfred.hofmann, ursula.barth, ingrid.haas, frank.holzwarth,|
\urldef{\mailsb}\path|anna.kramer, leonie.kunz, christine.reiss, nicole.sator,|
\urldef{\mailsc}\path|erika.siebert-cole, peter.strasser, lncs}@springer.com|

\newtheorem{theorem}{Theorem}

\newtheorem{corollary}{Corollary}
\newtheorem{definition}{Definition}

\newtheorem{remark}{Remark}

\begin{document}

\title{Characterization for a generic construction of bent functions and its consequences}

\author{Yanjun Li, Jinjie Gao\thanks{Corresponding author
\newline\indent Y. Li is with  Institute of Statistics and Applied Mathematics, Anhui University of Finance and Economics,  Bengbu, Anhui 233030, China; and School of Computer Sciences, Fudan University, Shanghai, 200433, China
(Email: yanjlmath90@163.com).
\newline \indent J. Gao is with Shanghai Key Laboratory of Intelligent Information Processing, School of Computer Science, Fudan University, Shanghai 200433, China; and Shanghai Engineering Research Center of Blockchain, Shanghai 200433, China (Email:  jjgao18@fudan.edu.cn).
\newline \indent H. Kan is with Shanghai Key Laboratory of Intelligent Information Processing, School of Computer Science, Fudan University, Shanghai 200433, China; Shanghai Engineering Research Center of Blockchain, Shanghai 200433, China;  and Yiwu Research Institute of Fudan University, Yiwu City 322000, China (E-mail: hbkan@fudan.edu.cn).
\newline \indent J. Peng is with Mathematics and Science College of
Shanghai Normal University, Guilin Road \#100, Shanghai 200234, China (Emails:~jpeng@shnu.edu.cn).
\newline \indent L. Zheng and C. Chen are  with the School of Mathematics and Physics, University of South China, Hengyang, Hunan 421001, China (Email: zhenglijing817@163.com  and cchxuexi@126.com).}~,
   Haibin Kan, Jie Peng,  Lijing Zheng, and Changhui Chen}
\maketitle
\begin{abstract}
 In this letter, we give a characterization for a generic construction of bent functions. This characterization enables us to obtain another efficient construction of bent functions and  to give a positive answer on a problem of bent functions.
\end{abstract}


\section{Introduction}

Bent functions, introduced in \cite{Rothaus}, are those Boolean functions in an even number of variables having the highest nonlinearity. Such functions have been extensively studied almost five decades, because of their closely relationship with the theory of difference sets, and their significant applications in coding theory and cryptography \cite{Cbook}. 
 In the past, a large amount of work was done on the characterizations and constructions of bent functions. But until now, a complete classification is not finished and it remains elusive. Along with the deep-going of the research, the progress on bent functions becomes more and more difficult, even if a tiny progress is not easy. For a comprehensive book on bent functions, the interested readers are referred to \cite{BookMesnager} for details.

In this letter, we focus our attention on the characterizations and constructions of bent functions with the form
\begin{align}\label{eqintro1}
h(x)=f(x)+F\circ \phi(x),
\end{align}
where $f$ is a bent function on $\mathbb{F}_{2^n}$, $F$ is a Boolean function on $\mathbb{F}_2^r$, and $\phi=(\phi_1,\phi_2,\ldots, \phi_r)$ is an $(n,r)$-function. In fact, the research on the bent-ness of $h$ can be dated back to \cite{Carlet-2006}, where Carlet presented a sufficient condition for a particular case  of $h$  to be bent. 
That sufficient condition had been proved by Mesnager \cite{Mesnager} to be necessary. 
 Mesnager \cite{Mesnager} also studied the bent-ness of two particular cases of Carlet function, 
 from which Mesnager obtained a lot of bent functions and gave their duals. Thereafter, several papers (such as \cite{Li et al.-2021, TangCM, WangLB, XuGK, ZhLJ1}) were done for generalizing Carlet's and Mesnager's works. 

In this letter, 
we obtain a characterization  for the generic construction of bent functions given in \cite{Li et al.-2021}, which 
enables us to find another efficient construction of bent functions. This characterization also enables us to provide a positive answer on the problem of bent function proposed in Conclusion of \cite{Li et al.-2021}.


\section{Preliminaries}\label{sec:Preliminaries}

 Throughout the paper, let $n=2m$ be an even positive integer. 
 Let $\mathbb{F}_{2^n}$ be the finite field of order $2^n$, $\mathbb{F}_{2^n}^*=\mathbb{F}_{2^n}\backslash\{0\}$, and $\mathbb{F}_2^n$ be the $n$-dimensional vector space over $\mathbb{F}_2$. 


For a vector $\omega=(\omega_1,\omega_2,\ldots,\omega_n)\in\mathbb{F}_2^n$, the set ${\rm suppt}(\omega)=\{1\leq i\leq n: \omega_i\neq 0\}$ is said to be the {\em support} of $\omega$, whose cardinality is called the {\em (Hamming) weight} of $\omega$, denoted by $wt(\omega)$. Namely, $wt(\omega)=|{\rm suppt}(\omega)|$.

A mapping $\phi$ from $\mathbb{F}_2^n$ to $\mathbb{F}_2^r$ is called an {\it$(n,r)$-function}. When $n$ is divisible by $r$, the $(n,r)$-function
\begin{align*}
{\rm Tr}_r^n(x)=x+x^{2^r}+x^{2^{2r}}+\cdots+x^{2^{n-r}}
\end{align*}
is called the \emph{trace function}. The set of all $(n,1)$-functions (namely, all {\it Boolean functions}) is denoted by $\mathcal{B}_n$.


For a given Boolean function $f$ on $\mathbb{F}_2^n$,
the {\em Walsh-Hadamard transform} of $f$  is a mapping from $\mathbb{F}_2^n$ to $\mathbb{Z}$ defined as
\begin{align*}
W_f(\mu)=\sum_{x\in\mathbb{F}_2^n}(-1)^{f(x)+\mu\cdot x}, \hspace{0.3cm} \mu\in\mathbb{F}_2^n,
\end{align*}
and its {\em inverse transform} is given by
 \begin{align*}
(-1)^{f(\mu)}=2^{-n}\sum_{x\in\mathbb{F}_2^n}W_f(x)(-1)^{\mu\cdot x}, \hspace{0.3cm} \mu\in\mathbb{F}_2^n,
\end{align*}
where $\mu\cdot x$ denotes the canonical  inner product of $\mu$ and $x$ (in $\mathbb{F}_{2^n}$, $\mu\cdot x={\rm Tr}_1^n(\mu x)$).

 The {\it first derivative} of $f$ in terms of $\mu\in\mathbb{F}_2^n$ is defined as
\begin{align*}
D_{\mu}f(x)=f(x)+f(x+\mu),
\end{align*}
and the {\it second derivative} of $f$ in terms of $\mu,\nu\in\mathbb{F}_2^n$ is defined as
\begin{align*}
D_{\mu}D_{\nu}f(x)=f(x)+f(x+\mu)+f(x+\nu)+f(x+\mu+\nu).
\end{align*}

\begin{definition}\label{defbent}
A Boolean function $f$ on $\mathbb{F}_2^n$ is called bent if $n$ is even and $W_f(\mu)=\pm2^{\frac{n}{2}}$ for all $\mu\in\mathbb{F}_2^n$.
\end{definition}

Bent functions always appear in pairs, that is, for any bent function $f$ on $\mathbb{F}_2^n$, there is always a unique bent function $f^*$ such that $W_f(\mu)=2^{\frac{n}{2}}(-1)^{f^*(\mu)}$ for all $\mu\in\mathbb{F}_2^n$ (in  literatures, $f^*$ is called the {\it dual} of $f$). 


\section{A characterization of a generic construction of  bent functions}\label{sec:frame}

In \cite{Li et al.-2021}, the authors have given a generic construction of bent functions, which generalizes the constructions of bent functions given in \cite{Carlet-2006, Mesnager, TangCM, WangLB, XuGK, ZhLJ1}. We restate it as follows.

\begin{theorem}\label{thlyj}\cite[Theorem 3]{Li et al.-2021}
Let $i$ be an integer with $1\leq i\leq r$,  $f,g_i\in\mathcal{B}_n$, and let $\phi=(\phi_1,\phi_2,\ldots,\phi_r)$ be the $(n,r)$-function with $\phi_i=f+g_i$. If the sum of any odd number of functions in $f,g_1,\ldots,g_r$ is a bent function, and its dual is equal to the sum of the duals of corresponding bent functions. Then for any Boolean function $F$ on $\mathbb{F}_2^r$, the function $h$ given by \eqref{eqintro1} is bent, and the dual of $h$ is
\begin{align*}
h^*(x)=f^*(x)+F\circ \varphi(x),
\end{align*}
where $\varphi=(\varphi_1,\varphi_2,\ldots,\varphi_r)$ is the $(n,r)$-function with $\varphi_i(x)=f^*(x)+g_i^*(x)$, $1\leq i\leq r$.
\end{theorem}
 
 Below, we want to generalize Theorem \ref{thlyj} by using the following property.

\begin{definition}[$\mathbf{P_r}$]\label{propty}
Let $f$ be a Boolean function over $\mathbb{F}_{2^n}$. If there is an $(n,r)$-function  $\phi=(\phi_1,\phi_2,\ldots, \phi_r)$ such that the following two conditions are satisfied:
\begin{enumerate}[label=(\roman*)]
 \item \label{item1}  $f(x)+\omega\cdot\phi(x)=f(x)+\sum_{i=1}^r\omega_i \phi_i$ is bent for any $\omega=(\omega_1,\omega_2,\ldots,\omega_r)\in\mathbb{F}_2^r$;

 \item \label{item2} there is an $(n,r)$-function $\varphi=(\varphi_1,\varphi_2,\ldots,\varphi_r)$ such that $\big(f(x)+\omega\cdot\phi(x)\big)^*=f^*(x)+\omega\cdot \varphi(x)$ for any $\omega\in\mathbb{F}_2^r$,
\end{enumerate}
then we say that $f$ satisfies \nameref{propty} with respect to the $(n,r)$-function $\phi$. 
\end{definition}


\begin{theorem}\label{thgbent}
Let $n=2m$. Let $\phi$ be an  $(n,r)$-function, and let $f$ be a Boolean function on $\mathbb{F}_{2^n}$ satisfying \nameref{propty} with respect to $\phi$. Then for any Boolean function $F$ on $\mathbb{F}_2^r$, the function $h$ given by \eqref{eqintro1}
is bent, and the dual of $h$ is
\begin{align*}
h^*(x)=f^*(x)+F\circ\varphi(x).
\end{align*}
\end{theorem}
\begin{proof}
By the definition of the inverse Walsh-Hadamard transform, it holds that
\begin{align*}
(-1)^{F\circ \phi(x)}=2^{-r}\sum_{\omega\in\mathbb{F}_2^r}W_F(\omega)(-1)^{\omega\cdot \phi(x)}, ~~\forall~x\in\mathbb{F}_2^n.
\end{align*}
Hence, the Walsh-Hadamard transform of $h$ at $\beta\in\mathbb{F}_{2^n}$ is that
\begin{align*}
W_h(\beta)=&\sum_{x\in\mathbb{F}_{2^n}}(-1)^{f(x)+{\rm Tr}_1^n(\beta x)}(-1)^{F\circ\phi(x)}\\
=&2^{-r}\sum_{x\in\mathbb{F}_{2^n}}(-1)^{f(x)+{\rm Tr}_1^n(\beta x)}\sum_{\omega\in\mathbb{F}_2^r}W_F(\omega)(-1)^{\omega \cdot \phi(x)}\\
=&2^{-r}\sum_{\omega\in\mathbb{F}_2^r}W_F(\omega)\sum_{x\in\mathbb{F}_{2^n}}(-1)^{f(x)+{\rm Tr}_1^n(\beta x)+\omega \cdot \phi(x)}\\
=&2^{-r}\sum_{\omega\in\mathbb{F}_2^r}W_F(\omega) W_g(\beta),
\end{align*}
where $g(x)=f(x)+\omega \cdot \phi(x)$. Recall that $f$ satisfies \nameref{propty} with respect to $\phi$, that is, $g$ is bent and $g^*(x)=f^*(x)+\omega\cdot \varphi(x)$ for any $\omega\in\mathbb{F}_2^r$. Hence, we have
\begin{align*}
W_h(\beta)=2^{m-r}\sum_{\omega\in\mathbb{F}_2^r}W_F(\omega) (-1)^{f^*(\beta)+\omega\cdot\varphi(\beta)}=2^{m}(-1)^{f^*(\beta)+F\circ\varphi(\beta)}.
\end{align*}
The proof is completed.
\end{proof}



\begin{corollary}\label{corlyj}
Theorem  \ref{thgbent} is reduced to that of Theorem \ref{thlyj} when $\phi=(\phi_1,\phi_2,\ldots, \phi_r)$ is an $(n,r)$-function with $\phi_i=f+g_i$,  $1\leq i\leq r$, where
 $f$ and $g_i$ are any Boolean functions on $\mathbb{F}_{2^n}$.
\end{corollary}
\begin{proof}
Suppose that  $\phi=(\phi_1,\phi_2,\ldots, \phi_r)$ with $\phi_i=f+g_i$, $1\leq i\leq r$. Then for any $\omega=(\omega_1,\omega_2,\ldots, \omega_r)\in\mathbb{F}_2^r$, we have
\begin{align*}
f(x)+\omega\cdot \phi(x)=f(x)+\sum_{i=1}^r\omega_i(f(x)+g_i(x))
=\begin{cases}
G_{\omega}(x),\hspace{0.3cm} &{\rm if}~~wt(\omega)~{\rm is~ odd},\\
f(x)+G_{\omega}(x),\hspace{0.3cm} &{\rm if}~~wt(\omega)~{\rm is~even},
\end{cases}
\end{align*}
where $G_{\omega}(x)=\omega_1g_1(x)+\omega_2g_2(x)+\cdots+\omega_rg_r(x)$. Therefore, Item \ref{item1} of \nameref{propty} holds if and only if the sum of any odd number of functions in $f,g_1,g_2,\ldots,g_r$ is bent. When ${\rm suppt}(\omega)=\{i\}$, we have
\begin{align*}
f(x)+\omega\cdot \phi(x)=g_i(x)\hspace{0.3cm}{\rm and} \hspace{0.3cm} f^*(x)+\omega\cdot \varphi(x)=f^*(x)+\varphi_i(x).
\end{align*}
So Item \ref{item2} of \nameref{propty} holds only if $\varphi_i(x)=f^*(x)+g_i^*(x)$ for any integer $i$, $1\leq i\leq r$. In this case,
 \begin{align*}
 f^*(x)+\omega\cdot \varphi(x)=f^*(x)+\sum_{i=1}^r\omega_i(f^*(x)+g_i^*(x))
 =\begin{cases}
G_{\omega}^*(x),\hspace{0.3cm} &{\rm if}~~wt(\omega)~{\rm is~ odd},\\
f^*(x)+G_{\omega}^*(x),\hspace{0.3cm} &{\rm if}~~wt(\omega)~{\rm is~even},
\end{cases}
 \end{align*}
 where $G_{\omega}^*(x)=\omega_1g_1^*(x)+\omega_2g_2^*(x)+\cdots+\omega_rg_r^*(x)$.
Hence, Item \ref{item2} of \nameref{propty} holds if and only if $(G_\omega)^*=G_{\omega}^*$ when $wt(\omega)$ is odd, and $(f+G_{\omega})^*=f^*+G_{\omega}^*$ when $wt(\omega)$ is even. Equivalently, the dual of the sum of any odd number of functions in $f,g_1,g_2,\ldots,g_r$ is equal to the sum of the duals of corresponding bent functions.
This completes the proof.
\end{proof}

From the proof of Corollary \ref{corlyj}, it is easily seen that for a given Boolean function $f$ on $\mathbb{F}_{2^n}$, and an $(n,r)$-function $\phi=(\phi_1,\phi_2,\ldots,\phi_r)$, \nameref{propty} holds if and only if the sum of any odd number of functions in $f, f+\phi_1,f+\phi_2,\ldots,f+\phi_r$ is bent, and its dual is equal to the sum of the duals of corresponding bent functions. Namely, Theorem \ref{thlyj} is another characterization of Theorem \ref{thgbent}. 
 Note that Theorem \ref{thlyj} was proved by induction in \cite{Li et al.-2021}. Here we provide a more simple alternative proof from another perspective.

Theorem \ref{thgbent} allows us to deduce the following result.

\begin{corollary}\label{cornew}
Let $n=2m$. Let $f$ and $g$ be two bent functions on $\mathbb{F}_{2^n}$. 
Let $\mu_2,\mu_3,\ldots,\mu_r\in\mathbb{F}_{2^n}^*$.
 If the following two conditions are satisfied:
\begin{enumerate}[label=(\Alph*)]
 \item \label{itemA}  $D_{\mu_i}D_{\mu_j}f^*=0$ for any  $2\leq i<j\leq r$;

 \item \label{itemB} for any  $\omega'=(\omega_2,\omega_3,\ldots,\omega_r)\in\mathbb{F}_2^{r-1}$, it holds that
 \begin{align}\label{eqitemB}
 g^*(x+\sum_{i=2}^r\omega_i\mu_i)
 =\begin{cases}
  g^*(x)+f^*(x)+\sum_{i=2}^r\omega_if^*(x+\mu_i),\hspace{0.2cm}&{\rm if}~wt(\omega'){\rm~ is~ odd},\\
  g^*(x)+\sum_{i=2}^r\omega_if^*(x+\mu_i),\hspace{0.2cm}&{\rm if}~wt(\omega'){\rm~ is~ even},
  \end{cases}
 \end{align}
\end{enumerate}
then for any Boolean function $F$ on $\mathbb{F}_2^r$, the function $h$ given by
\begin{align*}
 h(x)=f(x)+F(f(x)+g(x), {\rm Tr}_1^n(\mu_2x), {\rm Tr}_1^n(\mu_3x),\ldots,{\rm Tr}_1^n(\mu_rx))
 \end{align*}
 is bent. Moreover, the dual of $h$ is
\begin{align*}
h^*(x)=f^*(x)+F(\varphi_1,\varphi_2,\ldots,\varphi_r),
 \end{align*}
 where $\varphi_1(x)=f^*(x)+g^*(x)$ and $\varphi_i(x)=f^*(x)+f^*(x+\mu_i)$ for any integer $2\leq i\leq r$.
\end{corollary}
\begin{proof}
Let $\phi=(\phi_1,\phi_2,\ldots,\phi_r)$ be the $(n,r)$-function with $\phi_1(x)=f(x)+g(x)$ and $\phi_i(x)={\rm Tr}_1^n(\mu_i x)$ for each  $2\leq i\leq r$. Then for any $\omega=(\omega_1,\omega_2,\ldots,\omega_r)\in\mathbb{F}_2^r$, it is easily seen that
\begin{align*}
f(x)+\omega\cdot \phi(x)=
\begin{cases}
f(x)+{\rm Tr}_1^n((\omega_2\mu_2+\omega_3\mu_3+\cdots+\omega_r\mu_r)x), \hspace{0.3cm}&{\rm if}~\omega_1=0,\\
g(x)+{\rm Tr}_1^n((\omega_2\mu_2+\omega_3\mu_3+\cdots+\omega_r\mu_r)x), \hspace{0.3cm}&{\rm if}~\omega_1=1.
\end{cases}
\end{align*}
This implies that Item \ref{item1} of \nameref{propty} is satisfied when $f$ and $g$ are bent. So we have
\begin{align*}
\big(f(x)+\omega\cdot \phi(x)\big)^*=
\begin{cases}
f^*(x+\omega_2\mu_2+\omega_3\mu_3+\cdots+\omega_r\mu_r), \hspace{0.3cm}&{\rm if}~\omega_1=0,\\
g^*(x+\omega_2\mu_2+\omega_3\mu_3+\cdots+\omega_r\mu_r), \hspace{0.3cm}&{\rm if}~\omega_1=1.
\end{cases}
\end{align*}
Note that when ${\rm suppt}(\omega)=\{i\}$, we have
\begin{align*}
f(x)+\omega\cdot \phi(x)
=\begin{cases}g(x),\hspace{0.2cm}&{\rm if}~i=1,\\
f(x)+{\rm Tr}_1^n(\mu_ix) \hspace{0.2cm}&{\rm otherwise},
\end{cases}
\hspace{0.3cm}
{\rm and} \hspace{0.3cm} f^*(x)+\omega\cdot \varphi(x)=f^*(x)+\varphi_i(x).
\end{align*}
Hence, Item \ref{item2} of \nameref{propty} holds only if $\varphi_1(x)=f^*(x)+g^*(x)$ and $\varphi_i(x)=f^*(x)+f^*(x+\mu_i)$ for any  $2\leq i\leq r$. In this case,
\begin{align*}
f^*(x)+\omega\cdot \varphi(x)
=\begin{cases}
f^*(x)+\sum_{i=2}^r\omega_i(f^*(x)+f^*(x+\mu_i)), \hspace{0.3cm}&{\rm if}~\omega_1=0,\\
g^*(x)+\sum_{i=2}^r\omega_i(f^*(x)+f^*(x+\mu_i)), \hspace{0.3cm}&{\rm if}~\omega_1=1.
\end{cases}
\end{align*}
Hence, Item \ref{item2} of \nameref{propty} holds if and only if the following two relations hold:
\begin{align}\label{eqcor51}
f^*(x+\omega_2\mu_2+\cdots+\omega_r\mu_r)&=f^*(x)+\sum_{i=2}^r\omega_i(f^*(x)+f^*(x+\mu_i))\nonumber\\
&=\begin{cases}
f^*(x)+\sum_{i=2}^r\omega_if^*(x+\mu_i), \hspace{0.3cm}&{\rm if}~wt(\omega')~{\rm is~even},\\
\sum_{i=2}^r\omega_if^*(x+\mu_i), \hspace{0.3cm}&{\rm if}~wt(\omega')~{\rm is~odd},
\end{cases}
\end{align}
and
\begin{align}\label{eqcor52}
g^*(x+\omega_2\mu_2+\cdots+\omega_r\mu_r)&=g^*(x)+\sum_{i=2}^r\omega_i(f^*(x)+f^*(x+\mu_i))\nonumber\\
&=\begin{cases}
g^*(x)+\sum_{i=2}^r\omega_if^*(x+\mu_i), \hspace{0.3cm}&{\rm if}~wt(\omega')~{\rm is~even},\\
g^*(x)+f^*(x)+\sum_{i=2}^r\omega_if^*(x+\mu_i), \hspace{0.3cm}&{\rm if}~wt(\omega')~{\rm is~odd},
\end{cases}
\end{align}
where $\omega'=(\omega_2,\omega_3,\ldots,\omega_r)$. By \cite[Lemma 3.3]{ZhLJ1}, we know that Relation \eqref{eqcor51} holds if and only if $D_{\mu_i}D_{\mu_j}f^*=0$ for any  $2\leq i<j\leq r$.
Then the result follows from Theorem \ref{thgbent} immediately.
\end{proof}



Note that Condition  \ref{itemB} of  Corollary \ref{cornew} is elusive when $r>2$. In the following corollary, we give a reduced form by applying  Corollary \ref{cornew} to $g(x)=f(x+\alpha)$ for some $\alpha\in\mathbb{F}_{2^n}^*$.

\begin{corollary}\label{correduced}
Let $f$ be a bent function on $\mathbb{F}_{2^n}$. Let $\alpha\in\mathbb{F}_{2^n}$ and $\mu_2,\mu_3,\ldots,\mu_r\in\mathbb{F}_{2^n}^*$ 
be such that $\alpha\in\big<\mu_2,\mu_3,\ldots,\mu_r\big>^{\bot}$ and $D_{\mu_i}D_{\mu_j}f^*=0$ for any  $2\leq i<j\leq r$. Then for any Boolean function $F$ on $\mathbb{F}_2^r$, the function
 \begin{align*}
 h(x)=f(x)+F(f(x)+f(x+\alpha), {\rm Tr}_1^n(\mu_2x), {\rm Tr}_1^n(\mu_3x),\ldots,{\rm Tr}_1^n(\mu_rx))
 \end{align*}
 is bent. Moreover, the dual of $h$ is
\begin{align*}
h^*(x)=f^*(x)+F(\varphi_1,\varphi_2,\ldots,\varphi_r),
 \end{align*}
 where $\varphi_1(x)={\rm Tr}_1^n(\alpha x)$ and $\varphi_i(x)=f^*(x)+f^*(x+\mu_i)$ for any integer  $2\leq i\leq r$.
\end{corollary}
\begin{proof}
Let $g(x)=f(x+\alpha)$. Then it easily seen that $g^*(x)=f^*(x)+{\rm Tr}_1^n(\alpha x)$, and then Relation \eqref{eqitemB} becomes that
 \begin{align*}
 f^*(x+\sum_{i=2}^r\omega_i\mu_i)
 =\begin{cases}
\sum_{i=2}^r\omega_if^*(x+\mu_i),\hspace{0.2cm}&{\rm if}~wt(\omega'){\rm~ is~ odd},\\
  f^*(x)+\sum_{i=2}^r\omega_if^*(x+\mu_i),\hspace{0.2cm}&{\rm if}~wt(\omega'){\rm~ is~ even},
  \end{cases}
 \end{align*}
since $\alpha\in\big<\mu_2,\mu_3,\ldots,\mu_r\big>^{\bot}$. Hence, Condition \ref{itemB} of Corollary \ref{cornew} is satisfied if and only if $D_{\mu_i}D_{\mu_j}f^*=0$ for any $2\leq i<j\leq r$ by \cite[Lemma 3.3]{ZhLJ1}. Then the result follows from Corollary \ref{cornew} directly.
\end{proof}

\begin{remark}
Note that though the conditions of $h$ to be bent in Corollary \ref{correduced} are similar as that of given in \cite[Theorem 3.5]{ZhLJ1} (in fact, Corollary \ref{correduced} is reduced to \cite[Theorem 3.5]{ZhLJ1} when $\alpha=0$), the corresponding bent functions  in Corollary \ref{correduced}  and \cite[Theorem 3.5]{ZhLJ1} can be EA-inequivalent. For instance, let $n=6$ and
\begin{align*}
f(x)=(x_1,x_2,x_3)\cdot (x_4,x_5,x_6).
\end{align*}
Let $\mu_2=(1,0,0,0,0,0)$, $\mu_3=(0,1,1,0,0,0)$. Then it is easy to check that $D_{\mu_2}D_{\mu_3}f^*=0$. Hence, by \cite[Theorem 3.5]{ZhLJ1}, we have that
$$h(x)=f(x)+F(\mu_2 \cdot x, \mu_3\cdot x )=f(x)+F(x_1,x_2+x_3)$$
 is bent for any Boolean function $F$ on $\mathbb{F}_2^2$; and by Corollary \ref{correduced}, we have that
 \begin{align*}
 \hat{h}(x)=f(x)+\hat{F}\big(f(x)+f(x+\alpha),\mu_2\cdot x,\mu_3\cdot x\big)=f(x)+\hat{F}\big(f(x)+f(x+\alpha),x_1, x_2+x_3\big)
 \end{align*}
 is bent for any $\alpha\in\big<\mu_2,\mu_3\big>^{\bot}$ and any Boolean function $\hat{F}$ on $\mathbb{F}_2^3$. These two bent functions can be clearly EA-inequivalent, since the algebraic degree of $h$ is 2, while the algebraic degree of $\hat{h}$ is 3 when $\alpha=\mu_3$ and $\hat{F}(x_1,x_2,x_3)=x_1x_2x_3$.
\end{remark}


In \cite{Li et al.-2021}, the authors have found two kinds of $f$ and $\phi$ satisfying the conditions of Theorem \ref{thlyj} (that is, $\mathbf{P_r}$ by the previous discussion) for constructing new bent functions. The first kind is to let $f$ be a bent function and $\phi$ be a linear $(n,r)$-function; and the second kind is to let $f$ and $f+\phi_i$ be some self-dual bent functions for each $1\leq i\leq r$. They also invited the readers to find more kinds of $f$ and $\phi$ for obtaining more classes of bent functions in Conclusion of  \cite{Li et al.-2021}. Note that
 we have found a method to find such kinds of $f$ and $\phi$ in Corollary \ref{correduced}.

Then similarly as the concrete bent functions obtained in  \cite{Li et al.-2021}, \cite{WangLB}, \cite{XuGK} and \cite{ZhLJ1}, by applying Corollary \ref{correduced} to the following three monomial bent functions
\begin{align*}
&f_1(x)={\rm Tr}_1^n(\lambda x^{2^t+1}),\\
&f_2(x)={\rm Tr}_1^{6k}(\lambda x^{2^{2k}+2^k+1})\\
&f_3(x)={\rm Tr}_1^{4k}(\lambda x^{2^{2k}+2^{k+1}+1}),
\end{align*}
 respectively; and to the following bent functions with Niho exponents
\begin{align*}
f_4(x)={\rm Tr}_1^m(x^{2^m+1})+{\rm Tr}_1^n\bigg(\sum_{i=1}^{2^{k-1}-1}x^{(2^m-1)\frac{i}{2^k}+1}\bigg),
\end{align*}
  one can also obtain certain concrete bent functions, since by Corollary \ref{correduced}, one only needs to determine the duals of $f_1, f_2, f_3, f_4$ (which have been done in \cite{Leander}, \cite{Li et al.-2021}, \cite{Li et al.-ieice} and \cite{Nihodual}, respectively),  find some elements $\mu_2,\mu_3,\ldots,\mu_r\in\mathbb{F}_{2^n}^*$ such that $D_{\mu_i}D_{\mu_j}f_e^*=0$ for any  $2\leq i<j\leq r$ and any  $1\leq e\leq 4$ (such elements exist by \cite{Mesnager}, \cite{TangCM}, \cite{WangLB}, \cite{XuGK}, \cite{ZhLJ1}),  and find some $\alpha\in\mathbb{F}_{2^n}$ such that $\alpha\in\big<\mu_2,\mu_3,\ldots,\mu_r\big>^{\bot}$. Here, the concrete results do not unfolded in details.

\section{Conclusion}\label{sec:conclusion}

This paper gave another characterization for the generic construction of bent functions given in \cite{Li et al.-2021}, which enabled us to obtain another efficient construction of bent functions and to provide a positive answer on the problem of bent functions proposed in Conclusion of \cite{Li et al.-2021}.

\section*{Acknowledgments}
This work was supported in part by the National Key Research and Development Program of China under Grant 2019YFB2101703, in part
by the National Natural Science Foundation of China under Grants 62302001, 62372221, 61972258, 62272107 and U19A2066,  in part by the  China Postdoctoral Science Foundation under Grant 2023M740714, in part by the
Innovation Action Plan of Shanghai Science and Technology under Grants  20511102200 and 21511102200, in part by the Key Research and Development Program of Guangdong Province  under Grant 2020B0101090001, in part by the Natural Science Foundation for the Higher Education Institutions of Anhui Province under Grant 2023AH050250.



\end{document}